\documentclass[12pt]{article}
\usepackage{longtable,supertabular}
\usepackage{amsmath}
\usepackage{amssymb}
\usepackage{latexsym}
\usepackage{a4wide,color}
\usepackage{multirow}
\usepackage{ntheorem}
\usepackage[square,numbers,sort]{natbib}
\usepackage[title]{appendix}
\usepackage{array}
\usepackage{makecell}
\usepackage{enumitem}
\usepackage{booktabs}

\usepackage{xcolor}
\usepackage[colorlinks=true,linkcolor=blue,citecolor=blue,urlcolor=blue,bookmarks=false]{hyperref}

\usepackage{float}

\newtheorem{definition}{Definition}
\newtheorem{theorem}{Theorem}
\newtheorem{corollary}{Corollary}
\newtheorem{proposition}{Proposition}
\newtheorem{lemma}{Lemma}

\newtheorem*{proof}{Proof.}

\newtheorem{example}{Example}
\newtheorem{remark}{Remark}

\newcommand{\F}{\mathbb{F}}

\newcommand{\ord}{\mathrm{ord}}

\newcommand{\bC}{{\mathbf{{C}}}}

\newcommand{\bc}{{\mathbf{c}}}

\newcommand{\bV}{{\mathbf{V}}}

\newcommand{\bD}{{\mathbf{D}}}

\newcommand{\qed}{\hfill $\square$}

\title{\bf 
	On the Minimum Distances of Some Families of Goppa Codes and BCH Codes
	}
\author{
	Yaqi Chen, Hao Chen, Cunsheng Ding and Huimin Lao
	\thanks{ 		
		The research of Hao Chen was supported by NSFC Grant 62032009. The research of C. Ding was supported by Hong Kong Research Grants Council under Grant No. 16301123. 
		The research of Huimin Lao was supported by the National Research Foundation, Singapore and Infocomm Media Development Authority under its Trust Tech Funding Initiative. Any opinions, findings and conclusions or recommendations expressed in this material are those of the author(s) and do not reflect the views of
		National Research Foundation, Singapore and Infocomm Media Development Authority.
		(Corresponding author: Hao Chen)
		Yaqi Chen is with the College of Cyber Security, Jinan University, Guangzhou, Guangdong Province, 510632, China. (e-mail: chenyq@stu.jnu.edu.cn)
		Hao Chen is with the College of Information Science and
		Technology, Jinan University, Guangzhou, Guangdong Province, 510632, China. (e-mail: haochen@jnu.edu.cn)
		Cunsheng Ding is with the Department of Computer Science and Engineering, The Hong Kong University of Science and Technology, Hong Kong, China. (e-mail: cding@ust.hk)
		Huimin Lao is with Strategic Centre for Research in Privacy-Preserving Technologies and Systems, Nanyang Technological University, Singapore. (e-mail: huimin.lao@ntu.edu.sg)  
		}
}

\begin{document}
	
	\maketitle
	\begin{abstract}      
        Goppa codes form an important class of alternant codes with wide applications in algebraic coding theory and code-based cryptography. Determining the true minimum distance of a Goppa code is a difficult problem. In this paper, we provide a necessary and sufficient criterion for a Goppa code to attain its designed distance $\delta=t+1$, where $t$ is the degree of the Goppa polynomial. As applications, we determine the minimum distances of several classes of $q$-ary Goppa codes. In particular, we prove the tightness of the improved lower bound for a class of wild Goppa codes, and extend the family with $G(x)=x^t+A$ from the binary case to arbitrary odd prime powers.
        
        We then specialize the criterion to the monomial case $G(x)=x^t$, which is equivalent to primitive BCH codes.
        This leads to several infinite families of primitive BCH codes with $d=\delta$, including the binary codes $\bC_{(2,2^m-1,9,1)}$ and $\bC_{(2,2^m-1,15,1)}$, the family $\bC_{(p,p^p-1,2p+2,1)}$ with an odd prime $p$ and the family $\bC_{(q,q^m-1,r\frac{q^m-1}{q-1}+1,1)}$ with $r\mid q-1$. In particular, we prove that the primitive BCH code
        $\bC_{(q,q^m-1,q^t+1,1)}$ has minimum distance $q^t+1$ under the condition $t\mid m$, improving the previously known condition
        $pt\mid m$.\\

	\end{abstract}
	
	\newpage
 
\section{Introduction}

\subsection{Background}
Goppa codes, introduced by Goppa in 1970 \cite{1970Goppa}, form a fundamental class of alternant codes and play a central role in both algebraic coding theory and code-based cryptography, including the McEliece cryptosystem \cite{1978McEliece}.

Throughout this paper, let $\F_q$ denote the finite field with $q$ elements and $\F_q^\ast$ its multiplicative group. 
For a vector $\mathbf{a}=(a_0,\dots,a_{n-1})\in\F_q^n$, define 
$\mathrm{supp}(\mathbf{a})=\{i:a_i\ne 0\}$ and 
$\mathrm{wt}(\mathbf{a})=|\mathrm{supp}(\mathbf{a})|$. 
The Hamming distance is $d(\mathbf{a},\mathbf{b})=\mathrm{wt}(\mathbf{a}-\mathbf{b})$, 
and the minimum distance of a code $\mathbf{C}\subseteq\F_q^n$ is
$d(\mathbf{C})=\min_{\mathbf{a}\ne\mathbf{b}} d(\mathbf{a},\mathbf{b})$. A $q$-ary linear code $\mathbf{C}$ with parameters $[n,k,d]_q$ is a $k$-dimensional linear subspace of $\F_q^n$ with minimum distance $d$.

Let $q$ be a prime power and let $\F_{q^m}$ be an extension of $\F_q$.
Let $L=\{\alpha_1,\dots,\alpha_n\}\subseteq \F_{q^m}$ be a set of distinct elements, and let $G(x)\in\F_{q^m}[x]$ be a polynomial of degree $t$ such that
$G(\alpha_i)\ne 0$ for all $1\le i\le n$.

\begin{definition}[Goppa codes]\label{D-goppa}\cite{1970Goppa}
	The $q$-ary Goppa code associated with $L$ and $G$ is defined as
	\begin{equation}\label{e-goppa}
		\Gamma_q(L, G) := \left\{ \mathbf{c} = (c_1, c_2, \dots, c_n) \in \mathbb{F}_q^n : \sum_{i=1}^n \frac{c_i}{x - \alpha_i} \equiv 0 \pmod{G(x)} \right\}.
	\end{equation}
	The set $L$ is called the support set, and $G(x)$ is called the Goppa polynomial. The length of $\Gamma_q(L, G)$ is $n=|L|$.
\end{definition}

\begin{lemma}[Goppa bound]\label{L-goppa}\cite{1970Goppa}
	The Goppa code $\Gamma_q(L,G)$ has length $n$, dimension $k\ge n-mt,$
	and minimum distance $d\ge t+1$.
	The value $\delta=t+1$ is called its designed distance.
	For squarefree binary Goppa polynomials, the classical lower bound for binary Goppa codes improves to $d\ge 2t+1$, and the corresponding designed distance is $\delta=2t+1$.
\end{lemma}

It is well known that there exist asymptotic families of Goppa codes reaching the Gilbert–Varshamov bound.
Although the Goppa bound provides a universal lower bound on the minimum distance,
determining the exact true minimum distance remains notoriously difficult for general Goppa codes.
Consequently, an important problem in this field is to characterize and construct infinite families of Goppa codes with \(d=\delta\).
Such families are valuable because their error-correcting performance is precisely known and guaranteed.

BCH codes are closely related to Goppa codes, and form an important subclass of cyclic codes \cite{1960BCH}. Both classes can be viewed as subfield subcodes of generalized Reed–Solomon codes.
Let $\bC_{(q,n,\delta,b)}$ denote the BCH code over $\mathbb F_q$ with length $n$, designed distance $\delta$, and offset $b$. The case $b=1$ corresponds to the narrow-sense BCH code, while $n=q^m-1$ corresponds to the primitive BCH code.
In particular, when $n=q^m-1$, $L=\{1,\alpha,\dots,\alpha^{n-1}\}$, and $G(x)=x^{\delta-1}$, the Goppa code $\Gamma_q(L,G)$ is equivalent to the primitive BCH code $\bC_{(q,q^m-1,\delta,1)}$.

It is known that $\bC_{(q,n,\delta_1,b)}=\bC_{(q,n,\delta_2,b)}$ may hold for distinct $\delta_1$ and $\delta_2$. 
The largest designed distance defining a BCH code is called the Bose distance $d_B$ \cite[p.~171]{HP}. Thus
$d\ge d_B\ge \delta,$
which gives an approximation to the true minimum distance.

\subsection{Related works and motivations}

Early work on the minimum distances of Goppa codes focused primarily on binary cases and on highly structured defining polynomials. 
In \cite{1995BS}, Bezzateev and Shekhunova presented a subclass of binary Goppa codes with minimum distance equal to the designed distance $2t+1$ when $G(x)=x^t+A$, $t\mid 2^m-1$ and $A$ is a $t$-th power in $\F_{2^m}^\ast$. 
In \cite{2008BS}, they investigated a chain of separable binary Goppa codes with determined minimum distance defined by
$G(x)=x^{t-1}+1, x^t+x, x^t+x+1, x^t+x^{t-1}+1, x^{t+1}+1,$
where $t=2^\ell$ and $L \subseteq \mathbb{F}_{q^{2\ell}}$. 
In \cite{2010BS}, this chain structure was extended to $q$-ary separable Goppa codes. 
More recently, \cite{2019BS} introduced totally decomposed cumulative Goppa codes with $G(x)=(x^{t+1}+1)^j,$ where $t=q^\ell$ and $1\le j\le q$, and determined their exact minimum distances.

In \cite{2014wild}, Couvreur et al. investigated wild Goppa codes defined by norm polynomials. 
For a polynomial $g(x)\in \F_{q^m}[x]$ of degree $r$ with no roots in $\F_{q^m}$, they proved the identity
$$
\Gamma_q(L,g^{q^{m-1}+\cdots+q})
=\Gamma_q(L,g^{q^{m-1}+\cdots+q+1})
$$
for every support $L\subseteq \F_{q^m}$. This equivalence yields improved designed parameters for the corresponding code and guarantees a lower bound
$d\ge r(q^{m-1}+\cdots+q+1)+1.$
However, whether this improved lower bound is tight has remained open in general.

In parallel, determining the exact minimum distance of BCH codes is also a difficult and longstanding problem.
Related bounds for cyclic codes include Weil--Serre type bounds; see, for example, \cite{GuneriOzbudak2008}. 
Only a limited number of infinite families with $d=\delta$ are known. 
For a recent summary, we refer to \cite{2026BCH}. 
A particularly notable example is the primitive BCH code $\bC_{(q,q^m-1,q^t+1,1)}.$
It was conjectured in \cite[Conjecture~2]{2015Ding} that the minimum distance of this code always equals its Bose distance
$$
d_B=\left\lfloor \frac{q^m-2}{q^{m-t}-1}\right\rfloor+1.
$$
When $t\le m/2$, this reduces to $d_B=q^t+1$. In \cite{2026BCH}, a criterion was developed for BCH codes with $d=\delta$, and it was shown that
$\bC_{(q,q^m-1,q^t+1,1)}$ has minimum distance $q^t+1$ under the condition
$m\equiv 0\pmod{pt},$
where $p$ is the characteristic of $\F_q$.
This naturally raises the question of whether the Goppa viewpoint can provide a more structural explanation of the criterion in \cite[Theorem 3.1]{2026BCH} and identify more families of BCH codes with $d=\delta$.

The main motivations of this paper are the following.
\begin{itemize}
	\item Determining the exact minimum distance of a Goppa code is difficult in general, and most known results rely on restrictive assumptions e.g., $t=q^\ell$ and $L\subseteq \mathbb F_{q^{2\ell}}$.
	\item For wild Goppa codes, the tightness of the lower bound on the minimum distance has not been settled.
	\item The close connection between Goppa codes and BCH codes suggests that new results on the minimum distance of BCH codes may be obtained from the Goppa viewpoint.
\end{itemize}

The aim of this paper is to address these problems and determine the exact minimum distances of several families of $q$-ary Goppa codes and primitive BCH codes from a unified viewpoint.

\subsection{Contributions and organization}
The main contributions of this paper are summarized as follows.
\begin{itemize}
	\item We obtain a necessary and sufficient criterion for a $q$-ary Goppa code to attain its designed distance, see Theorem~\ref{T-0}. In the case $G(x)=x^t$ and $L=\F_{q^m}^*$, this criterion recovers the criterion obtained in \cite[Theorem 3.1]{2026BCH} for BCH codes.

	\item Three families of Goppa codes with $d=\delta$ are presented. In particular, we prove the tightness of the improved bound for wild Goppa codes when $r\mid q-1$ and $r>1$, and extend the Goppa codes with $G(x)=x^t+A$ in \cite{1995BS} from binary cases to arbitrary odd prime powers. 
	
	\item We derive four families of narrow-sense primitive BCH codes with $d=\delta$, including $\bC_{(2,2^m-1,\delta,1)}$ with $\delta \in \{9,15\}$, 
	$\bC_{(p,p^p-1,2p+2,1)}$ where $p$ is an odd prime,
	$\bC_{(q, q^m-1, t+1, 1)}$ with $r \mid q-1$ and $t=r\frac{q^m-1}{q-1}$.
	In particular, we show that $\bC_{(q,q^m-1,q^t+1,1)}$ has minimum distance $q^t+1$ when $t\mid m$, improving the requirement $pt \mid m$ in \cite[Theorem 4.4]{2026BCH}.

\end{itemize}

The code families with $d=\delta$ studied in this paper are summarized in Table~\ref{tab:main-results}. 
For each Goppa family therein, the support set is
$L=\{\alpha\in\F_{q^m}:G(\alpha)\neq 0\}.$

\begin{table}[H]
	\centering
	\caption{Code classes studied in this paper}
	\label{tab:main-results}
	\renewcommand{\arraystretch}{1.15}
	\small
	\begin{tabular}{|p{1cm}|p{3.2cm}|p{5cm}|p{3.2cm}|p{1.6cm}|}
		\hline
		\centering Type & \centering $G(x)$ / code class & \centering Main assumptions & \centering Minimum distance & \centering Reference \tabularnewline
		\hline
		\multirow{7}{*}{\makecell{Goppa\\codes}}
		& $N_{\F_{q^m}/\F_q}(g(x))$
		& $\deg(g)=r>1$, $r\mid q-1$, $g$ has no roots in $\F_{q^m}$
		& $d=r\frac{q^m-1}{q-1}+1$
		& Thm.~\ref{T-4-2} \\
		\cline{2-5}
		& $x^t+A$
		& $q$ is odd, $t\mid q^m-1$, $A$ is a $t$-th power in $\F_{q^m}^\ast$
		& $d=t+1$
		& Thm.~\ref{T-4-3} \\
		\cline{2-5}
		& $(x+u)^t-\lambda(x+v)^t$
		& $t+1\mid(q^m-1)$, $p\nmid t+1$, $u\neq v$, $\lambda\in\F_{q^m}^\ast\setminus\{1\}$ is a $(t+1)$-th power
		& $d=t+1$
		& Thm.~\ref{T-4-4} \\
		\hline
		\multirow{5}{*}{\makecell{BCH\\codes}}
		& $\bC_{(2,2^m-1,\delta,1)}$
		& $\delta=9$ with $8\mid m$; $\delta=15$ with $14\mid m$
		& $d=\delta$
		& Thm.~\ref{T-q=2} \\
		\cline{2-5}
		& $\bC_{(p,p^p-1,2p+2,1)}$
		& $p$ is an odd prime
		& $d=2p+2$
		& Thm.~\ref{T-5-1} \\
		\cline{2-5}
		& $\bC_{(q,q^m-1,r\frac{q^m-1}{q-1}+1,1)}$
		& $r\mid q-1$, $1\le r<q-1$
		& $d=r\frac{q^m-1}{q-1}+1$
		& Thm.~\ref{T-BCH} \\
		\cline{2-5}
		& $\bC_{(q,q^m-1,q^t+1,1)}$
		& $t\mid m$, $t<m$
		& $d=q^t+1$
		& Thm.~\ref{T-q^t+1} \\
		\hline
	\end{tabular}
\end{table}

The rest of this paper is organized as follows. 
Section~\ref{s2} recalls the basic definitions and preliminaries on Goppa codes and BCH codes. 
Section~\ref{s3} proves a criterion for Goppa codes to have $d=\delta$, derives constructive consequences, and specializes the criterion to primitive BCH codes. 
Section~\ref{s4} applies this criterion to determine the minimum distance of three classes of Goppa codes. 
Section~\ref{s5} derives four explicit families of narrow-sense primitive BCH codes with $d=\delta$. 
Section~\ref{s6} concludes the paper.
\section{Preliminaries}\label{s2}

In this section, we recall the definitions and some known results of Goppa codes and BCH codes.

The wild Goppa codes studied in \cite{2014wild} are defined as follows.

\begin{definition}\cite{2014wild}\label{D-wild}
	Let $g(x) \in \mathbb{F}_{q^m}[x]$ be a polynomial of degree $r$. The $q$-ary wild Goppa code associated with the norm of $g(x)$ is defined by the Goppa polynomial$$G(x) = \text{N}_{\mathbb{F}_{q^m}/\mathbb{F}_q}(g(x)) = g(x)^{\frac{q^m-1}{q-1}},$$where $L \subseteq \mathbb{F}_{q^m}$ is chosen such that $G(\alpha) \neq 0$ for all $\alpha \in L$.
\end{definition}

We next recall the definition of narrow-sense BCH codes.

\begin{definition}
	Let $q$ be a prime power and let $n$ be a positive integer such that $\gcd(n,q)=1$.  Let $m=\ord_n(q)$, and let $\alpha$ be a primitive $n$-th root of unity in $\mathbb{F}_{q^m}$. 
	For an integer $\delta\ge 2$, the narrow-sense BCH code $\bC_{(q,n,\delta,1)}$ is the cyclic code of length $n$ over $\mathbb{F}_q$ with generator polynomial
	$$
	g(x)=\mathrm{lcm}\bigl(m_1(x),m_2(x),\dots,m_{\delta-1}(x)\bigr),
	$$
	where $m_i(x)$ is the minimal polynomial of $\alpha^i$ over $\mathbb{F}_q$. Equivalently, $g(x)$ is the monic polynomial of least degree over $\F_q$ having
	$\alpha,\alpha^2,\dots,\alpha^{\delta-1}$
	as zeros. The integer $\delta$ is called the designed distance.
\end{definition}

The following two lemmas give the dimensions of certain BCH codes.

\begin{lemma}\label{L-dim1}\cite{AA2007}
	Suppose $q^{\lceil m/2 \rceil} < n \le q^m - 1$, where $m=\ord_n(q)$. The narrow-sense BCH code 
	$\bC_{(q,n,\delta,1)}$ with $\delta$ in the range $2 \le \delta \le \min \{\lfloor  n q^{\lceil m/2 \rceil}/(q^m - 1) \rfloor , n  \}$ has dimension
	$$
	k=n-m\lceil (\delta-1)(1-1/q)  \rceil.
	$$
\end{lemma}

\begin{lemma}[Theorem 4.10. \cite{CJL2020}]\label{L-dim2}
	Let $t=r\frac{q^m-1}{q-1}$
	for some integer $r$ with $1\le r<q-1$. Then the narrow-sense primitive BCH code $\bC_{(q,q^m-1,t+1,1)}$
	has dimension $k=(q-r)^m-1.$
\end{lemma}

The following lemma records the well-known equivalence between a Goppa code and a primitive BCH code. A brief proof is included for completeness.

\begin{lemma}\label{L-2-1}
	Let $n = q^m - 1$ and $\alpha$ be a primitive element of $\mathbb{F}_{q^m}$. Let the support set be the ordered tuple $L = \{1, \alpha, \alpha^2, \dots, \alpha^{n-1}\}$ and the Goppa polynomial be $G(x) = x^{\delta-1}$. Then the Goppa code $\Gamma_q(L, G)$ is equivalent to the primitive BCH code $\bC_{(q, n, \delta, 1)}$.
\end{lemma}

\begin{proof}
\rm{
	Let $t=\delta-1$ and write $\alpha_i=\alpha^i$ for $0\le i\le n-1$. By Definition \ref{D-goppa}, a parity-check matrix of $\Gamma_q(L,G)$ is
	$$
	H=
	\begin{pmatrix}
		G(\alpha_0)^{-1} & G(\alpha_1)^{-1} & \cdots & G(\alpha_{n-1})^{-1}\\
		\alpha_0G(\alpha_0)^{-1} & \alpha_1G(\alpha_1)^{-1} & \cdots & \alpha_{n-1}G(\alpha_{n-1})^{-1}\\
		\vdots & \vdots & \ddots & \vdots\\
		\alpha_0^{t-1}G(\alpha_0)^{-1} & \alpha_1^{t-1}G(\alpha_1)^{-1} & \cdots & \alpha_{n-1}^{t-1}G(\alpha_{n-1})^{-1}
	\end{pmatrix}.
	$$
	Since $G(x)=x^t$, the entry in the $r$-th row and $i$-th column is
	$$
	(\alpha^i)^{r-1}(\alpha^i)^{-t}=\alpha^{i(r-t-1)},
	\quad 1\le r\le t.
	$$
	Hence
	$$
	H=
	\begin{pmatrix}
		1 & \alpha^{-t} & \alpha^{-2t} & \cdots & \alpha^{-(n-1)t}\\
		1 & \alpha^{-(t-1)} & \alpha^{-2(t-1)} & \cdots & \alpha^{-(n-1)(t-1)}\\
		\vdots & \vdots & \vdots & \ddots & \vdots\\
		1 & \alpha^{-1} & \alpha^{-2} & \cdots & \alpha^{-(n-1)}
	\end{pmatrix}.
	$$
	Let $\beta=\alpha^{-1}$. Then $\beta$ is also a primitive element of $\mathbb{F}_{q^m}$. Reversing the order of the rows of $H$ yields
	$$
	H'=
	\begin{pmatrix}
		1 & \beta & \beta^2 & \cdots & \beta^{n-1}\\
		1 & \beta^2 & \beta^4 & \cdots & \beta^{2(n-1)}\\
		\vdots & \vdots & \vdots & \ddots & \vdots\\
		1 & \beta^t & \beta^{2t} & \cdots & \beta^{t(n-1)}
	\end{pmatrix},
	$$
	which is the standard parity-check matrix of the primitive BCH code $\bC_{(q,n,t+1,1)}$ with consecutive roots
	$\beta,\beta^2,\dots,\beta^t.$
	Therefore, $\Gamma_q(L,G)$ is equivalent to $\bC_{(q,n,t+1,1)}$.
	\qed
}
\end{proof}

\section{Characterization of Goppa codes with $d=\delta$ and applications to primitive BCH codes}\label{s3}

In this section, we establish a necessary and sufficient condition for a Goppa code to
have minimum distance $d=\delta=t+1$ and derive a constructive result of this criterion.
Then, we specialize the monomial case $G(x)=x^t$ to primitive BCH codes and obtain a sufficient condition for $\bC_{(q,q^m-1,\delta,1)}$ to have minimum distance $d=\delta$.

\subsection{A criterion for Goppa codes with $d=\delta$}
\begin{theorem}\label{T-0}
	Let $\Gamma_q(L,G)$ be a Goppa code with $\deg(G)=t$. Then the minimum distance of $\Gamma_q(L,G)$ is equal to $t+1$ if and only if there exist pairwise distinct elements
	$
	\alpha_{i_1},\alpha_{i_2},\ldots,\alpha_{i_{t+1}}\in L
	$
	such that, for
	$
	F(x)=\prod_{\ell=1}^{t+1}(x-\alpha_{i_\ell}),
	$
	the following holds
	\begin{equation}\label{e-T-0-0}
		R_j:=\frac{G(\alpha_{i_j})}{G(\alpha_{i_{t+1}})}
		\cdot
		\frac{F'(\alpha_{i_{t+1}})}{F'(\alpha_{i_j})}
		\in \mathbb{F}_q^\ast,
		\quad 1\le j\le t.
	\end{equation}

\end{theorem}

\begin{proof}
\rm{
	We first prove the sufficiency. Assume that there exist pairwise distinct elements $\alpha_{i_1},\ldots,\alpha_{i_{t+1}}\in L$ satisfying \eqref{e-T-0-0}. By the Goppa bound $d \ge t+1$, it suffices to construct a codeword of weight $t+1$.
	
	Let
	$
	I=\{i_1,i_2,\ldots,i_{t+1}\}\subseteq \{1,2,\ldots,n\},
	$
	and let $\bc=(c_1,c_2,\ldots,c_n)\in \mathbb{F}_q^n$ be a vector with
	$
	\operatorname{supp}(\bc)=I.
	$
	Write $c_{i_1},c_{i_2},\ldots,c_{i_{t+1}}$ for its nonzero coordinates. By Definition \ref{D-goppa}, the vector $\bc$ belongs to $\Gamma_q(L,G)$ if and only if
	\begin{equation}\label{e-T-0-1}
		\sum_{\ell=1}^{t+1}\frac{c_{i_\ell}}{x-\alpha_{i_\ell}}\equiv 0 \pmod{G(x)}.
	\end{equation}
	Equivalently, the coordinates $c_{i_1},\ldots,c_{i_{t+1}}$ satisfy
	\begin{equation}\label{e-T-0-2}
		\sum_{\ell=1}^{t+1} c_{i_\ell}\frac{\alpha_{i_\ell}^r}{G(\alpha_{i_\ell})}=0,
		\quad 0\le r\le t-1.
	\end{equation}
	
	Separate the variable $c_{i_{t+1}}$. Let
	$$
	\bV=
	\begin{pmatrix}
		1 & 1 & \cdots & 1\\
		\alpha_{i_1} & \alpha_{i_2} & \cdots & \alpha_{i_t}\\
		\vdots & \vdots & \ddots & \vdots\\
		\alpha_{i_1}^{t-1} & \alpha_{i_2}^{t-1} & \cdots & \alpha_{i_t}^{t-1}
	\end{pmatrix}
	$$
	and
	$$
	\bD=\operatorname{diag}\left(G(\alpha_{i_1})^{-1},\ldots,G(\alpha_{i_t})^{-1}\right).
	$$
	Then the coefficient matrix of $c_{i_1},\ldots,c_{i_t}$ in \eqref{e-T-0-2} is $\bV\bD$. Since the elements $\alpha_{i_1},\ldots,\alpha_{i_t}$ are distinct and $G(\alpha_{i_j})\neq 0$, the matrix $\bV\bD$ is invertible. 
	Hence
	\begin{equation}\label{e-T-0-3}
		\begin{pmatrix}
			c_{i_1}\\
			c_{i_2}\\
			\vdots\\
			c_{i_t}
		\end{pmatrix}
		=
		-\frac{c_{i_{t+1}}}{G(\alpha_{i_{t+1}})}
		\bD^{-1}\bV^{-1}
		\begin{pmatrix}
			1\\
			\alpha_{i_{t+1}}\\
			\vdots\\
			\alpha_{i_{t+1}}^{t-1}
		\end{pmatrix}.
	\end{equation}
	
	By the Lagrange interpolation formula, the $j$-th entry of
	$$
	V^{-1}
	\begin{pmatrix}
		1\\
		\alpha_{i_{t+1}}\\
		\vdots\\
		\alpha_{i_{t+1}}^{t-1}
	\end{pmatrix}
	$$
	is
	$$
	\prod_{\substack{l=1 \\ l \neq j}}^{t} \frac{\alpha_{i_{t+1}} - \alpha_{i_l}}{\alpha_{i_j} - \alpha_{i_l}}.
	$$
	
	On the other hand,
	$$
		\prod_{\substack{l=1 \\ l \neq j}}^{t} (\alpha_{i_{t+1}} - \alpha_{i_l}) = \frac{F'(\alpha_{i_{t+1}})}{\alpha_{i_{t+1}} - \alpha_{i_j}}, 
		\quad \quad 
		\prod_{\substack{l=1 \\ l \neq j}}^{t} (\alpha_{i_j} - \alpha_{i_l}) = \frac{F'(\alpha_{i_j})}{\alpha_{i_j} - \alpha_{i_{t+1}}},
	$$
	and therefore
	\begin{equation}\label{e-T-0-4}
		\prod_{\substack{l=1 \\ l \neq j}}^{t} \frac{\alpha_{i_{t+1}} - \alpha_{i_l}}{\alpha_{i_j} - \alpha_{i_l}}
		=-\frac{F'(\alpha_{i_{t+1}})}{F'(\alpha_{i_j})}.
	\end{equation}
	Substituting \eqref{e-T-0-4} into \eqref{e-T-0-3} gives
	\begin{equation}\label{e-T-0-5}
		c_{i_j}=
		c_{i_{t+1}}
		\cdot
		\frac{G(\alpha_{i_j})}{G(\alpha_{i_{t+1}})}
		\cdot
		\frac{F'(\alpha_{i_{t+1}})}{F'(\alpha_{i_j})},
		\qquad 1\le j\le t.
	\end{equation}

	By assumptions, $$R_j=\frac{G(\alpha_{i_j})}{G(\alpha_{i_{t+1}})}
	\cdot
	\frac{F'(\alpha_{i_{t+1}})}{F'(\alpha_{i_j})} $$ belongs to $\mathbb{F}_q^\ast$ for every $j$. Hence, after choosing any $c_{i_{t+1}}\in \mathbb{F}_q^\ast$, all coordinates $c_{i_j}$ also belong to $\mathbb{F}_q^\ast$. This yields a valid codeword of weight $t+1$ in $\Gamma_q(L,G)$. Therefore, $d=t+1$.
	
	We next prove the necessity. Suppose that $d=t+1$. Then there exists a codeword
	$
	\bc=(c_1,c_2,\ldots,c_n)\in \Gamma_q(L,G)
	$
	with $\operatorname{wt}(\bc)=t+1$. 
	Let
	$
	\operatorname{supp}(\bc)=\{i_1,i_2,\ldots,i_{t+1}\},
	$
	and let $\alpha_{i_1},\ldots,\alpha_{i_{t+1}}\in L$ be the corresponding support elements. Since $\bc$ is a codeword, its nonzero coordinates satisfy \eqref{e-T-0-2}, and the same computation gives \eqref{e-T-0-5}. Since $c_{i_j},c_{i_{t+1}}\in \mathbb{F}_q^\ast$, we obtain
	$$
	\frac{G(\alpha_{i_j})}{G(\alpha_{i_{t+1}})}
	\cdot
	\frac{F'(\alpha_{i_{t+1}})}{F'(\alpha_{i_j})}=
	\frac{c_{i_j}}{c_{i_{t+1}}}
	\in \mathbb{F}_q^\ast,
	\qquad 1\le j\le t.
	$$
	Thus \eqref{e-T-0-0} holds. This completes the proof. 
\qed
}
\end{proof}

\begin{remark}
\rm{
	Theorem~\ref{T-0} is specific to the case $d=\delta=t+1$. 
	For larger weights, the same approach leads to a Vandermonde-type system involving additional unknowns, which in general does not yield a closed-form criterion depending only on the support elements. 
	Therefore, the characterization in Theorem~\ref{T-0} is naturally restricted to this case.
}
\end{remark}

Theorem~\ref{T-0} also admits a constructive interpretation. 
In particular, once a set of $t+1$ distinct elements $\alpha_{i_1},\ldots,\alpha_{i_{t+1}}\in \F_{q^m}$ satisfies \eqref{e-T-0-0}, a Goppa code with minimum distance $d=t+1$ can be constructed. 
This leads to the following corollary.

\begin{corollary}\label{C-derivative}
	Let $q$ be a power of a prime $p$, and $t$ an integer such that $p\nmid t+1$.
	Let
	$F(x)\in \mathbb F_{q^m}[x]$
	be a polynomial of degree $t+1$ with $t+1$ distinct roots
	$\alpha_{i_1},\ldots,\alpha_{i_{t+1}}\in \mathbb F_{q^m}$.
	Define
	$$
	G(x)=(t+1)^{-1}F'(x).
	$$
	Then $\deg(G)=t$ and
	$G(\alpha_{i_j})\neq 0$ for all $1\le j\le t+1$.
	Moreover, for any support set
	$L\subseteq \mathbb F_{q^m}$
	containing $\{\alpha_{i_1},\ldots,\alpha_{i_{t+1}}\}$ and satisfying $G(\alpha) \neq 0$ for all $\alpha \in L$,
	the Goppa code $\Gamma_q(L,G)$ has minimum
	distance $d=t+1$.
\end{corollary}

\begin{proof}
	\rm{
		Since $F(x)$ has $t+1$ distinct roots and $p\nmid t+1$, the polynomial $F'(x)$ has degree
		$t$ and $F'(\alpha_{i_j})\neq 0$ for all $j$. Hence $\deg(G)=t$ and
		$G(\alpha_{i_j})\neq 0$ for all $j$.
		
		For $1\le j\le t$, we obtain
		$$
		\frac{G(\alpha_{i_j})}{G(\alpha_{i_{t+1}})}
		\cdot
		\frac{F'(\alpha_{i_{t+1}})}{F'(\alpha_{i_j})}
		=
		1\in \mathbb F_q^\ast.
		$$
		Therefore the condition \eqref{e-T-0-0} in Theorem~\ref{T-0} is satisfied, and the conclusion follows.
		\qed
	}
\end{proof}

\begin{remark}
	\rm{
		Corollary~\ref{C-derivative} shows that Theorem~\ref{T-0} is constructive rather than merely existential.
		More generally, let $w_1,\ldots,w_{t+1} \in \mathbb{F}_q^\ast$, and define
		$$
		G(x) = \sum_{\ell=1}^{t+1} w_\ell \frac{F(x)}{x-\alpha_{i_\ell}}.
		$$
		Then $G(x)$ is a polynomial of degree $t$. Moreover,
		$$
		G(\alpha_{i_\ell}) = w_\ell F'(\alpha_{i_\ell}) \neq 0,\quad 1\le \ell \le t+1.
		$$
		
		Consequently, for any support set $L \subseteq \mathbb{F}_{q^m}$ containing $\{\alpha_{i_1},\ldots,\alpha_{i_{t+1}}\}$ and satisfying $G(\alpha)\neq 0$ for all $\alpha\in L$, the Goppa code $\Gamma_q(L,G)$ has minimum distance $d=t+1$.
		Corollary~\ref{C-derivative} corresponds to the special case $w_1=\cdots=w_{t+1}=(t+1)^{-1}$.
	}
\end{remark}

\subsection{Specialization to primitive BCH codes}

We now specialize Theorem~\ref{T-0} to the case when $G(x)=x^t$ and $L=\{1,\alpha,\alpha^2,\dots,\alpha^{n-1}\}$ where $\alpha$ is a primitive element of $\F_{q^m}$. 
By Lemma~\ref{L-2-1}, the Goppa code $\Gamma_q(L,G)$ is equivalent to the primitive BCH code $\bC_{(q,q^m-1,t+1,1)}$. 
A criterion for primitive BCH codes to have $d=\delta$ was obtained in \cite[Theorem~3.1]{2026BCH}. 
The next proposition shows that this criterion is precisely the monomial specialization of Theorem~\ref{T-0}.

\begin{proposition}\label{P-BCH}
	Let $n=q^m-1$, let $\alpha$ be a primitive element of $\F_{q^m}$.
	Define
	$L=\{1,\alpha,\alpha^2,\dots,\alpha^{n-1}\}$ and $G(x)=x^t.$
	Under the equivalence in Lemma~\ref{L-2-1}, Theorem~\ref{T-0} is equivalent to the locator criterion in \cite[Theorem~3.1]{2026BCH}.
\end{proposition}

\begin{proof}
\rm{
	Let
	$$
	\alpha_{i_{t+1}}=1
	\quad\text{and}\quad
	\alpha_{i_j}=x_j^{-1},
	\quad 1\le j\le t,
	$$
	and define
	$$
	F(x)=(x-1)\prod_{k=1}^t (x-x_k^{-1}).
	$$
	Then
	$$
	F(x)=\prod_{\ell=1}^{t+1}(x-\alpha_{i_\ell}).
	$$
	Since $G(x)=x^t$, a direct computation gives
	$$
	\frac{G(\alpha_{i_j})}{G(\alpha_{i_{t+1}})}
	\cdot
	\frac{F'(\alpha_{i_{t+1}})}{F'(\alpha_{i_j})}
	=
	-\frac{\prod_{k\ne j}(1-x_k)}
	{x_j\prod_{k\ne j}(x_j-x_k)},
	\quad 1\le j\le t.
	$$
	Write
	$$
	S_j:=
	\frac{\prod_{k\ne j}(1-x_k)}
	{x_j\prod_{k\ne j}(x_j-x_k)}.
	$$
	Then
	$$
	\frac{G(\alpha_{i_j})}{G(\alpha_{i_{t+1}})}
	\cdot
	\frac{F'(\alpha_{i_{t+1}})}{F'(\alpha_{i_j})}
	=
	-S_j,
	\quad 1\le j\le t.
	$$
	
	On the other hand, equation \eqref{e-T-0-5} in the proof of Theorem~\ref{T-0} yields
	$$
	c_{i_j}
	=
	c_{i_{t+1}}
	\cdot
	\frac{G(\alpha_{i_j})}{G(\alpha_{i_{t+1}})}
	\cdot
	\frac{F'(\alpha_{i_{t+1}})}{F'(\alpha_{i_j})},
	\quad 1\le j\le t.
	$$
	Hence
	$$
	c_{i_j}=-c_{i_{t+1}}S_j,
	\quad 1\le j\le t.
	$$
	Therefore, the condition in Theorem~\ref{T-0} is equivalent to the  criterion in \cite[Theorem~3.1]{2026BCH}.
	\qed
}
\end{proof}

The following corollary gives a sufficient condition such that the primitive BCH code
$\bC_{(q,q^m-1,\delta,1)}$ attains its designed distance.

\begin{corollary}\label{C-5-1}
	Let $\delta=t+1$. Let $M(x)\in \F_{q^m}[x]$ be a squarefree polynomial of degree $t$
	such that the following hold
	\begin{enumerate}[label=(\roman*)]
		\item $M(x)$ has $\delta-1$ distinct nonzero roots $x_1,\ldots,x_t\in \F_{q^m}$.
		\item $x(x-1)M'(x)\equiv w \pmod{M(x)}$ for some $w\in \F_q^*$.
		\item $M(1)\in \F_q^*$.
	\end{enumerate}
	Then the primitive BCH code $\bC_{(q,q^m-1,\delta,1)}$ has minimum distance $\delta$.
\end{corollary}

\begin{proof}
	\rm{
	Let $x_1,\ldots,x_t$ be the roots of $M(x)$. By Proposition \ref{P-BCH}, it suffices to show that
	$$
	S_j:=
	\frac{\prod_{k\ne j}(1-x_k)}
	{x_j\prod_{k\ne j}(x_j-x_k)}
	\in \F_q^*,
	\qquad 1\le j\le t.
	$$
	Since
	$$
	M(x)=\prod_{k=1}^t (x-x_k),
	$$
	we have
	$$
	M(1)=(1-x_j)\prod_{k\ne j}(1-x_k)
	\quad\text{and}\quad
	M'(x_j)=\prod_{k\ne j}(x_j-x_k).
	$$
	Hence
	$$
	S_j=
	\frac{M(1)}
	{x_j(1-x_j)M'(x_j)}
	=
	-\frac{M(1)}
	{x_j(x_j-1)M'(x_j)}.
	$$
	By (ii), evaluating the congruence at $x=x_j$ gives
	$$
	x_j(x_j-1)M'(x_j)=w.
	$$
	Since $M(1)\in \F_q^*$ and $w\in \F_q^*$, it follows that
	$S_j=-\frac{M(1)}{w}\in \F_q^*.$
	\qed
	}
\end{proof}

The next lemma gives a coefficient criterion for the congruence condition (ii) in Corollary~\ref{C-5-1}.
\begin{lemma}\label{L-5-coeff}
	Let
	$
	M(x)=x^t+a_{t-1}x^{t-1}+\cdots+a_1x+a_0\in \F_q[x]
	$
	be monic with $a_0\ne 0$, and let $w\in \F_q^*$. Then the following are equivalent
	\begin{enumerate}[label=(\roman*)]
		\item
		$x(x-1)M'(x)\equiv w \pmod{M(x)}$.
		\item There exists $b\in \F_q$ such that
		\begin{equation}\label{e-C-2-0}
			\begin{cases}
				(t-k+1)a_{k-1}+(k+b)a_k=0, & 2\le k\le t,\\[3pt]
				ta_0+(b+1)a_1=0,\\[3pt]
				w=-ba_0.
			\end{cases}
		\end{equation}
	\end{enumerate}
\end{lemma}

\begin{proof}
\rm{
	Assume that {\rm (i)} holds. Then
	$
	x(x-1)M'(x)-w
	$
	is divisible by $M(x)$, it follows that
	$$
	x(x-1)M'(x)=(tx+b)M(x)+w
	$$
	for some $b\in \F_q$.
	
	Comparing the coefficients of $x^k$ for $2\le k\le t$, of $x$, and of the constant term, respectively, we obtain
	$$
	(t-k+1)a_{k-1}+(k+b)a_k=0,\qquad 2\le k\le t,
	$$
	$$
	ta_0+(b+1)a_1=0,
	\qquad
	w=-ba_0.
	$$
	Hence {\rm (ii)} holds.
	
	Conversely, assume that {\rm (ii)} holds. Then
	$$
	x(x-1)M'(x)=(tx+b)M(x)+w,
	$$
	and hence
	$$
	x(x-1)M'(x)\equiv w \pmod{M(x)}.
	$$
	So {\rm (i)} and {\rm (ii)} are equivalent.
	\qed
}
\end{proof}
\section{Goppa codes with $d=\delta$}\label{s4}

In this section, we apply Theorem \ref{T-0} to several families of Goppa codes and prove that their minimum distances equal the designed distance.

\subsection{Wild Goppa codes}

Throughout this subsection, assume that $m\ge 2$. We begin with the wild Goppa codes of Definition~\ref{D-wild}. The starting point is the following result of Couvreur et al.~\cite{2014wild}.

\begin{proposition}[\cite{2014wild}]\label{P-1}
	Let $g(x)\in \mathbb{F}_{q^m}[x]$ be a polynomial of degree $r$ with no roots in $\mathbb{F}_{q^m}$, and define
	$$
	N(g)(x)=N_{\mathbb{F}_{q^m}/\mathbb{F}_q}(g(x))
	=g(x)^{1+q+\cdots+q^{m-1}}.
	$$
	Then, for every support set $L\subseteq \mathbb{F}_{q^m}$,
	$$
	\Gamma_q(L,N(g))=\Gamma_q(L,N(g)/g).
	$$
	Consequently, this common code satisfies
	$$
	d \ge r\frac{q^m-1}{q-1}+1.
	$$
\end{proposition}

In particular, for
$G(x)=N_{\mathbb F_{q^m}/\mathbb F_q}(g(x)),$
we have $\deg(G)=r\frac{q^m-1}{q-1}.$
Accordingly, throughout this subsection we put
$t=r\frac{q^m-1}{q-1}.$
Proposition \ref{P-1} gives the lower bound $d\ge t+1$. We now show that this lower bound is tight when $r>1$ and $r\mid q-1$.

The proof uses the following auxiliary polynomial.

\begin{lemma}\label{L-H}

	Assume that $r\mid q-1$, and let
	$$
	F_\gamma(x)=(x+\gamma)^{t+1}-(x+\gamma),
	\quad \gamma\in\mathbb F_{q^m}.
	$$
	Then $F_\gamma(x)$ has exactly $t+1$ distinct roots in $\mathbb{F}_{q^m}$. Moreover, for every root $\alpha$ of $F_\gamma(x)$, we have
	$F_\gamma'(\alpha)\in\mathbb F_q^\ast.$

\end{lemma}

\begin{proof}
	\rm{
		Let $p$ be the characteristic of $\mathbb F_q$. 
		Since $t=r\frac{q^m-1}{q-1}$ and $r\mid q-1$, it follows that $t\mid q^m-1$. Hence the equation $y^t=1$ has exactly $t$ distinct solutions in $\mathbb F_{q^m}^\ast$. The roots of $F_\gamma(x)$ are precisely the solutions of
		$$
		x+\gamma=0
		\quad\text{or}\quad
		(x+\gamma)^t=1,
		$$
		and thus $F_\gamma(x)$ has exactly $t+1$ distinct roots in $\mathbb F_{q^m}$.
		
		Next,
		$$
		F_\gamma'(x)=(t+1)(x+\gamma)^t-1.
		$$
		If $\alpha=-\gamma$, then $F_\gamma'(\alpha)=-1\in\mathbb F_q^\ast$. 
		If $\alpha\neq-\gamma$ is a root, then $(\alpha+\gamma)^t=1$, and hence $F_\gamma'(\alpha)=t$.
		
		Note that
		$$
		\frac{q^m-1}{q-1}=1+q+\cdots+q^{m-1}\equiv 1 \pmod{p}.
		$$
		Thus $\frac{q^m-1}{q-1}\in\mathbb F_q^\ast$. Since $r\mid q-1$, we have $p\nmid r$, and therefore
		$
		t=r\frac{q^m-1}{q-1}\in\mathbb F_q^\ast.
		$
		This completes the proof.
		\qed
	}
\end{proof}
\begin{theorem}\label{T-4-2}
	Let $q$ be a prime power. Let $g(x)\in \mathbb F_{q^m}[x]$ be of degree $r>1$ with no roots in $\mathbb F_{q^m}$, and set
	$G(x)=N_{\mathbb F_{q^m}/\mathbb F_q}(g(x))$.
	If $r\mid q-1$ and $L=\mathbb F_{q^m}$, then the wild Goppa code $\Gamma_q(L,G)$ has minimum distance
	$d=t+1=r\frac{q^m-1}{q-1}+1.$
	
\end{theorem}

\begin{proof}
\rm{
	By Proposition \ref{P-1}, the code $\Gamma_q(L,G)$ has minimum distance at least $t+1$. It remains to show that it contains a codeword of weight $t+1$.
	
	Choose any $\gamma\in\mathbb F_{q^m}$, and let
	$\{\alpha_{i_1},\alpha_{i_2},\ldots,\alpha_{i_{t+1}}\}$
	be the set of roots of
	$F_\gamma(x)=(x+\gamma)^{t+1}-(x+\gamma)$
	from Lemma \ref{L-H}. Since $L=\mathbb F_{q^m}$, all these elements belong to the support set. We apply Theorem \ref{T-0} with
	$$
	F(x)=F_\gamma(x)=\prod_{\ell=1}^{t+1}(x-\alpha_{i_\ell}).
	$$
	
	For every $\alpha\in\mathbb F_{q^m}$, we obtain
	$$
	G(\alpha)=g(\alpha)^{1+q+\cdots+q^{m-1}}\in\mathbb F_q.
	$$
	Because $g$ has no roots in $\mathbb F_{q^m}$, it follows that
	$G(\alpha)\in\mathbb F_q^\ast$ for all $\alpha\in\mathbb F_{q^m}.$
	Hence
	\begin{equation}\label{e-T-2-0}
		\frac{G(\alpha_{i_j})}{G(\alpha_{i_{t+1}})}\in\mathbb F_q^\ast,
		\quad 1\le j\le t.
	\end{equation}
	
	On the other hand, Lemma \ref{L-H} gives
	$F'(\alpha_{i_j}) \in\mathbb F_q^\ast$ for $1\le j\le t+1$
	and therefore
	\begin{equation}\label{e-T-2-1}
		\frac{F'(\alpha_{i_{t+1}})}{F'(\alpha_{i_j})}\in\mathbb F_q^\ast,
		\quad 1\le j\le t.
	\end{equation}

	Combining \eqref{e-T-2-0} and \eqref{e-T-2-1}, we obtain
	$$
	\frac{G(\alpha_{i_j})}{G(\alpha_{i_{t+1}})}
	\cdot
	\frac{F'(\alpha_{i_{t+1}})}{F'(\alpha_{i_j})}
	\in\mathbb F_q^\ast,
	\quad 1\le j\le t.
	$$
	Thus the condition \eqref{e-T-0-0} in Theorem \ref{T-0} is satisfied. Hence $\Gamma_q(L,G)$ contains a codeword of weight $t+1$, and consequently
	$d=t+1=r\frac{q^m-1}{q-1}+1.$
	This completes the proof. 
\qed
}
\end{proof}

\begin{remark}
\rm{
	The case $r=1$ is excluded here, since Proposition \ref{P-1} requires that $g$ has no roots in $\mathbb F_{q^m}$. The boundary case $r=1$ will be treated separately in Section \ref{s5}.
	
}
\end{remark}

\begin{example}
	\rm{
		Examples of the codes in Theorem \ref{T-4-2} are given in Table \ref{tab-wild}. 
		All numerical
		examples were produced using SageMath.
	}
\end{example}

\begin{longtable}{|l|l|l|l|l|}
	\caption{\label{tab-wild} Codes in Theorem \ref{T-4-2}}\\ 
	\hline 
	$q$ & $m$ & $r$ & $\Gamma_q(L,N(g))$ & $d_{best}$ \\  
	\hline
	\endfirsthead
	
	\hline
	$q$ & $m$ & $r$ & $\Gamma_q(L,N(g))$ & $d_{best}$ \\  
	\hline
	\endhead
	
	\hline
	\endfoot
	
	5 & 2 & 2 & $[25,9,13]_5$  &  13 \\ \hline

	\multirow{2}{*}{7} 
	& \multirow{2}{*}{2} & 2 & $[49,25,17]_7$   &  17 \\ \cline{3-5}
	&  & 3 & $[49,16,25]_7$  &  25 \\ \hline
	
	\multirow{2}{*}{9} 
	& \multirow{2}{*}{2} & 2 & $[81,49,21]_9$   &  21 \\ \cline{3-5}
	&  & 4 & $[81,25,41]_9$  &  41 \\ \hline

\end{longtable}

\subsection{Two explicit families of  Goppa codes with $d=\delta$}

In \cite{1995BS}, Bezzateev and Shekhunova proved that the binary Goppa codes with polynomial $G(x)=x^t+A$ have minimum distance $2t+1$ when $t\mid2^m-1$ and $A$ is a $t$-th power in $\F_{2^m}^\ast$. We generalize this result to $q$-ary Goppa codes  when $q$ is an odd prime power. 

\begin{theorem}\label{T-4-3}
	Let $q$ be a power of an odd prime $p$, and $m$ and $t$ be positive integers such that $t\mid q^m-1$. Assume that $A$ is a $t$-th power in $\mathbb F_{q^m}^\ast$. Let $G(x)=x^t+A$ and $L=\{\alpha\in\mathbb F_{q^m}:G(\alpha)\neq 0\}$. Then the $q$-ary Goppa code $\Gamma_q(L,G)$ has minimum distance $d=t+1$.
\end{theorem}

\begin{proof}
\rm{
	Since $t\mid q^m-1$ and $A$ is a $t$-th power in $\mathbb F_{q^m}^\ast$, the equation $x^t=A$ has exactly $t$ distinct solutions in $\mathbb F_{q^m}^\ast$. Let $\alpha_{i_1},\ldots,\alpha_{i_t}$ be these solutions, and set $\alpha_{i_{t+1}}=0$. Define $$F(x)=x(x^t-A).$$ 
	Then $F(x)$ has exactly the $t+1$ distinct roots $\alpha_{i_1},\ldots,\alpha_{i_t},\alpha_{i_{t+1}}$.
	
	It remains to verify that these roots belong to $L$. Clearly, $G(0)=A\neq 0$. For $1\le j\le t$, the relation $\alpha_{i_j}^t=A$ gives $G(\alpha_{i_j})=2A$. As $q$ is odd, $2\in\mathbb F_q^\ast$, and hence $G(\alpha_{i_j})\neq 0$. Thus all roots of $F(x)$ belong to $L$.
	
	Since $F(x)=\prod_{\ell=1}^{t+1}(x-\alpha_{i_\ell})$, we have $F'(x)=(t+1)x^t-A$. Hence $F'(0)=-A$ and $F'(\alpha_{i_j})=tA$ for $1\le j\le t$. Therefore, for each $1\le j\le t$,
	$$
	\frac{G(\alpha_{i_j})}{G(\alpha_{i_{t+1}})}
	\cdot
	\frac{F'(\alpha_{i_{t+1}})}{F'(\alpha_{i_j})}
	=
	\frac{2A}{A}\cdot\frac{-A}{tA}
	=
	-\frac{2}{t}.
	$$
	Since $p\nmid q^m-1$ and $t\mid q^m-1$, we  have $t\in\mathbb F_q^\ast$. Hence $-\frac{2}{t}\in\mathbb F_q^\ast$. The condition \eqref{e-T-0-0} in Theorem~\ref{T-0} is satisfied. Therefore $\Gamma_q(L,G)$ contains a codeword of weight $t+1$, and thus $d=t+1$.\qed
}
\end{proof}

\begin{example}
	\rm{
		Examples of the codes in Theorem \ref{T-4-3} are listed in Table \ref{tab-q-goppa}.
		All numerical results were generated using SageMath.
		When $t$ is small, the minimum distances of the codes $\Gamma_q(L,G)$ are often close to the best known values recorded in the codetables \cite{codetable}.
	}
\end{example}

\begin{longtable}{|c|c|c|c|c|}
	\caption{\label{tab-q-goppa} Codes in Theorem \ref{T-4-3}}\\
	\hline
	$q$ & $m$ & $t$ & $\Gamma_q(L,G)$ & $d_{\mathrm{best}}$ \\
	\hline
	\endfirsthead
	
	\hline
	$q$ & $m$ & $t$ & $\Gamma_q(L,G)$ & $d_{\mathrm{best}}$ \\
	\hline
	\endhead
	
	\hline
	\endfoot
	
	\multirow{5}{*}{3}
	& 2 & 2  & $[7,3,3]_3$    & 4  \\ \cline{2-5}
	& 3 & 2  & $[27,21,3]_3$  & 4  \\ \cline{2-5}
	& \multirow{3}{*}{4} 
	& 4  & $[77,61,5]_3$  & 8  \\ \cline{3-5}
	&     & 8  & $[73,42,9]_3$  & 14 \\ \cline{3-5}
	&     & 10 & $[71,37,11]_3$ & 15 \\ \hline
	
	\multirow{6}{*}{5}
	& \multirow{4}{*}{2}
	& 2 & $[23,19,3]_5$  & 4  \\ \cline{3-5}
	&  & 4  & $[21,13,5]_5$ & 6 \\ \cline{3-5}
	&     & 6 & $[19,10,7]_5$  & 7  \\ \cline{3-5}
	&     & 8 & $[25,10,9]_5$  & 12 \\ \cline{2-5}
	
	& \multirow{2}{*}{3}
	& 2 & $[123,117,3]_5$  & 4  \\ \cline{3-5}
	&     & 4 & $[125,113,5]_5$  & 6 \\ \hline
	
	\multirow{5}{*}{7}
	& \multirow{5}{*}{2}
	& 2  & $[47,43,3]_7$  & 4  \\ \cline{3-5}
	&     & 3  & $[49,43,4]_7$  & 4  \\ \cline{3-5}
	&     & 4  & $[45,37,5]_7$  & 6  \\ \cline{3-5}
	&     & 8  & $[41,28,9]_7$  & 9  \\ \cline{3-5}
	&     & 12 & $[37,14,13]_7$ & 16 \\ \hline
	
	\multirow{4}{*}{9}
	& \multirow{4}{*}{2}
	& 2  & $[79,75,3]_9$  & 4  \\ \cline{3-5}
	&     & 4  & $[77,69,5]_9$  & 6  \\ \cline{3-5}
	&     & 8  & $[73,57,9]_9$  & 10  \\ \cline{3-5}
	&     & 10  & $[71,54,11]_9$  & 11  \\ \hline
\end{longtable}

We now present the second family obtained from a fractional linear parametrization of the roots.

\begin{theorem}\label{T-4-4}
	Let $q$ be a power of a prime $p$, and $m,t$ be positive integers such that $t+1\mid(q^m-1)$ and $p\nmid t+1$. Assume that  $u,v\in\mathbb F_{q^m}$ with $u\neq v$ and $\lambda\in\mathbb F_{q^m}^\ast\setminus\{1\}$ is a $(t+1)$-th power. Let $G(x)=(x+u)^t-\lambda(x+v)^t$ and
	$L=\{\alpha\in\mathbb F_{q^m}:G(\alpha)\neq 0\}.$
	Then the Goppa code $\Gamma_q(L,G)$ has minimum distance $d=t+1$.
\end{theorem}

\begin{proof}
\rm{
	Choose $\eta\in\mathbb F_{q^m}^\ast$ such that $\eta^{t+1}=\lambda$, and define
	$$
	F(x)=(x+u)^{t+1}-\lambda(x+v)^{t+1}.
	$$
	
	Since $u\neq v$, we have $F(-v)=(u-v)^{t+1}\neq 0$, so every root $x$ of $F(x)$ satisfies $x\neq -v$. For such $x$, set
	$$
	y=\frac{x+u}{x+v}.
	$$
	Then
	$$
	F(x)=0 \iff y^{t+1}=\lambda.
	$$
	
	Since $t+1\mid(q^m-1)$ and $\lambda$ is a $(t+1)$-th power in $\mathbb F_{q^m}^\ast$, the equation $y^{t+1}=\lambda$ has exactly $t+1$ solutions. The fractional linear transformation $x\mapsto \frac{x+u}{x+v}$ is invertible, hence $F(x)$ has exactly $t+1$ distinct roots in $\mathbb F_{q^m}$.
	
	Moreover,
	$$
	F'(x)=(t+1)\big((x+u)^t-\lambda(x+v)^t\big)=(t+1)G(x),
	$$
	so $G(x)=(t+1)^{-1}F'(x)$. Since $p\nmid t+1$, every root $\alpha$ of $F(x)$ satisfies $G(\alpha)\neq 0$, and hence all roots lie in the support set $L$.
	Therefore, by Corollary \ref{C-derivative}, the code $\Gamma_q(L,G)$ has minimum distance
	$d=t+1.$
\qed}
\end{proof}

\begin{example}
	\rm{
		Examples of the codes in Theorem \ref{T-4-4} are given in Table \ref{tab-T-4-4}. 
		All numerical
		examples were produced using SageMath.
	}
\end{example}

\begin{longtable}{|c|c|c|c|c|}
	\caption{\label{tab-T-4-4} Codes in Theorem \ref{T-4-4}}\\
	\hline
	$q$ & $m$ & $t$ & $\Gamma_q(L,G)$ & $d_{\mathrm{best}}$ \\
	\hline
	\endfirsthead
	
	\hline
	$q$ & $m$ & $t$ & $\Gamma_q(L,G)$ & $d_{\mathrm{best}}$ \\
	\hline
	\endhead
	
	\hline
	\endfoot
	
	\multirow{8}{*}{3}
	& \multirow{2}{*}{2}
	& 1  & $[8,6,2]_3$    & 2  \\ \cline{3-5}
	&     & 3  & $[8,4,4]_3$    & 4  \\ \cline{2-5}
	& \multirow{2}{*}{3}
	& 1  & $[26,23,2]_3$   & 2  \\ \cline{3-5}
	&     & 12 & $[27,4,13]_3$   & 18 \\ \cline{2-5}
	& \multirow{4}{*}{4}
	& 1  & $[80,76,2]_3$   & 2  \\ \cline{3-5}
	&     & 3  & $[80,72,4]_3$   & 4  \\ \cline{3-5}
	&     & 4  & $[81,65,5]_3$   & 8  \\ \cline{3-5}
	&     & 7  & $[80,52,8]_3$   & 12 \\ \hline
	
	\multirow{4}{*}{4}
	& \multirow{2}{*}{2}
	& 2  & $[15,11,3]_4$   & 4  \\ \cline{3-5}
	&     & 4  & $[15,9,5]_4$    & 5  \\ \cline{2-5}
	& \multirow{2}{*}{3}
	& 2  & $[63,57,3]_4$   & 4  \\ \cline{3-5}
	&     & 6 & $[64,46,7]_4$  & 9 \\ \hline
	
	\multirow{7}{*}{5}
	& \multirow{4}{*}{2}
	& 1  & $[24,22,2]_5$   & 2  \\ \cline{3-5}
	&     & 3  & $[25,19,4]_5$   & 5  \\ \cline{3-5}
	&     & 5  & $[24,16,6]_5$   & 8  \\ \cline{3-5}
	&     & 7  & $[24,10,8]_5$   & 11 \\ \cline{2-5}
	& \multirow{3}{*}{3}
	& 1  & $[124,121,2]_5$ & 2  \\ \cline{3-5}
	&     & 3  & $[124,115,4]_5$ & 5  \\ \cline{3-5}
	&     & 30 & $[125,53,31]_5$ & 38 \\ \hline

	\multirow{5}{*}{7}
	& \multirow{5}{*}{2}
	& 1 & $[48,46,2]_7$  & 2 \\ \cline{3-5}
	&  & 2 & $[49,45,3]_7$  & 3 \\ \cline{3-5}
	&  & 3 & $[49,43,4]_7$  & 4 \\ \cline{3-5}
	&  & 5 & $[48,38,6]_7$  & 6 \\ \cline{3-5}
	&  & 7 & $[48,36,8]_7$  & 8 \\ \hline
	
	\multirow{2}{*}{8}
	& \multirow{2}{*}{2}
	& 2 & $[63,59,3]_8$ & 4 \\ \cline{3-5}
	&  & 6 & $[64,52,7]_8$ & 8 \\ \hline
	
	\multirow{4}{*}{9}
	& \multirow{4}{*}{2}
	& 1 & $[80,78,2]_9$ & 2 \\ \cline{3-5}
	&  & 3 & $[80,74,4]_9$ & 4 \\ \cline{3-5}
	&  & 4 & $[81,73,5]_9$ & 5 \\ \cline{3-5}
	&  & 7 & $[80,66,8]_9$ & 8 \\ \hline
\end{longtable}

\section{BCH codes with $d=\delta$}\label{s5}

In this section, we apply the results in Section \ref{s3} to several families of primitive BCH codes and prove that their minimum distances equal the designed distance.

\subsection{Binary BCH codes with $\delta \in \{9,15\}$}

\begin{theorem}\label{T-q=2}
	Let $m \ge 8$. The binary BCH code $\bC_{(2,2^m-1,\delta,1)}$ has the following parameters:
	\begin{itemize}
		\item For any $8 \mid m$, the BCH code $\bC_{(2,2^m-1,9,1)}$ has parameters $[2^m-1,2^m-4m-1,9]_2$.
		\item For any $14 \mid m$, the BCH code $\bC_{(2,2^m-1,15,1)}$ has parameters $[2^m-1,2^m-7m-1,15]_2$.
	\end{itemize}
\end{theorem}

\begin{proof}
\rm{
	The dimensions follow from Lemma \ref{L-dim1}. It remains to prove the minimum distance.
	\begin{itemize}
		\item Consider the polynomial $L_1(x)=x^8 + x^7 + x^2 + x + 1$. It is verified by SageMath that $L_1(x)$ is irreducible over $\F_2$.
		Since $x(x-1)L_1'(x) \equiv 1 \pmod{L_1(x)}$, it follows from Corollary \ref{C-5-1} that the BCH code $\bC_{(2,2^m-1,9,1)}$ has minimum distance $9$ for any $8 \mid m$.
		\item Consider the polynomial $L_2(x)=x^{14} + x^{13} + x^6 + x^5 + 1$. It is verified by SageMath that $L_2(x)$ is irreducible over $\F_2$.
		Since $x(x-1)L_2'(x) \equiv 1 \pmod{L_2(x)}$, it follows from Corollary \ref{C-5-1} that the BCH code $\bC_{(2,2^m-1,15,1)}$ has minimum distance $15$ for any $14 \mid m$.\qed
	\end{itemize}
	}
\end{proof}

\subsection{$p$-ary BCH codes with $\delta=2p+2$}
In this subsection, we prove that the BCH code $\bC_{(p,p^p-1,2p+2,1)}$ has minimum distance $2p+2$ for any odd prime $p$.
The following lemmas will be used.

\begin{lemma}\label{L-5-1}
	Let $p$ be a prime.
	Then the polynomial $H(x)=x^p+x^{p-1}-1 \in \F_p[x]$ is irreducible over $\F_p$.
\end{lemma}

\begin{proof}
\rm{
	The reciprocal polynomial of $H(x)$ is
	$$
	H^*(x)=x^p[(\frac{1}{x})^p + (\frac{1}{x})^{p-1}-1]
	=-(x^p-x-1). 
	$$
	It is well known that the Artin--Schreier polynomial $x^p-x-1$ is irreducible over $\F_p$.
	Hence $H^*(x)$ is irreducible over $\F_p$. 
	Since irreducibility is preserved under taking reciprocal polynomials, $H(x)$ is also irreducible over $\F_p$.
	\qed
}
\end{proof}

\begin{lemma}\label{L-5-2}
	Let $p$ be a prime.
	Then the polynomial $M(y)=y^p+y^{p-1}+ \cdots +y-1 \in \F_p[y]$ is irreducible over $\F_p$.
\end{lemma}

\begin{proof}
\rm{
	Set $y=x+1$. Since $M(y)=\frac{y^{p+1}-1}{y-1}-2$, we obtain
	$$
	M(x+1)=\frac{(x+1)^{p+1}-1}{(x+1)-1}-2.
	$$
	In characteristic $p$, $(x+1)^{p+1}=x^{p+1}+x^p+x+1$. Therefore,
	$$
	M(x+1)=x^p+x^{p-1}-1.
	$$
	By Lemma~\ref{L-5-1}, the polynomial $M(x+1)$ is irreducible over $\F_p$. 
	Since irreducibility is invariant under translations, $M(y)$ is also irreducible over $\F_p$.
	\qed
}
\end{proof}

\begin{lemma}\label{L-5-3}
	Let $p$ be an odd prime, and define
	$$
	L(x)=x^{2p+1}+x^{2p}+\cdots+x^2-x-1 \in \F_p[x].
	$$
	Then the splitting field of $L(x)$ over $\F_p$ is $\F_{p^p}$.
\end{lemma}

\begin{proof}
\rm{
	Write
	$$
	L(x)=(x+1)\bigl(x^{2p}+x^{2p-2}+\cdots+x^2-1\bigr)
	=(x+1)M(x^2),
	$$
	where
	$
	M(y)=y^p+y^{p-1}+\cdots+y-1.
	$
	By Lemma~\ref{L-5-2}, the polynomial $M(y)$ is irreducible over $\F_p$ of degree $p$. Hence its splitting field is $\F_{p^p}$.
	
	By the change of variable $y=x+1$, the roots of $M(y)$ correspond to the roots of $x^p+x^{p-1}-1$. 
	Equivalently, if $\beta$ satisfies $\beta^p-\beta-1=0$, then
	$$
	\alpha=1+\beta^{-1}.
	$$
	Since $\beta^p=\beta+1$, we obtain
	$$
	\alpha=\frac{\beta+1}{\beta}=\beta^{p-1}.
	$$
	
	Since $p$ is odd, $p-1$ is even, and hence $\alpha=\beta^{p-1}$ is a square in $\F_{p^p}$. 
	Therefore every root of $M(y)$ has two square roots in $\F_{p^p}$, and thus all roots of $M(x^2)$ lie in $\F_{p^p}$.

	Conversely, if $\gamma$ is a root of $M(x^2)$, then $\gamma^2$ is a root of $M(y)$. 
	Hence the splitting field of $L(x)$ contains $\F_{p^p}$. 
	Combined with the previous inclusion, it follows that the splitting field of $L(x)$ is  $\F_{p^p}$. 
	\qed
}
\end{proof}

\begin{theorem}\label{T-5-1}
	Let $p$ be an odd prime.
	The BCH code $\bC_{(p,p^p-1,2p+2,1)}$ has minimum distance $2p+2$ and dimension $p^p-1-2p^2+p$.
\end{theorem}

\begin{proof}
\rm{
	The dimension follows from Lemma \ref{L-dim1}. It remains to prove the minimum distance.
	
	Let $L(x)=x^{2p+1}+x^{2p}+\cdots+x^2-x-1 \in \F_p[x]$. Then $L(0)=-1 \neq 0$. By Lemma \ref{L-5-3}, the polynomial $L(x)$ splits over $\F_{p^p}$ and has $2p+1$ distinct nonzero roots. 
	In particular, $L(x)$ is squarefree.

	Next, we verify that $L(x)$ satisfies the conditions of Corollary \ref{C-5-1}.
	By direct computation, it can be checked that the coefficients of $L(x)$ satisfy the relations in Lemma \ref{L-5-coeff} with $b=-2$, which yields
	$$
	x(x-1)L'(x) \equiv -2 \pmod{L(x)}.
	$$
	Moreover, $L(1)=-2 \in \F_p^*.$
	Therefore Corollary \ref{C-5-1} applies with $q=p$, $m=p$, $\delta=2p+2$, and $w=-2$. Hence the primitive BCH code $\bC_{(p,p^p-1,2p+2,1)}$ has minimum distance $2p+2$.
\qed
}
\end{proof}

\begin{example}
	\rm{
		In the following examples, let $\alpha$ be a primitive element of $\F_{p^p}$.  
		All computations were confirmed by SageMath.
		\begin{itemize}
			\item Let $p=3$. 
			The polynomial $L(x)=x^7 + x^6 + x^5 + x^4 + x^3 + x^2 + 2x + 2$ has roots $\{\alpha,\alpha^3,\alpha^9,\alpha^{13},\alpha^{14},\alpha^{16},\alpha^{22}\}$ in $\F_{3^3}$.
			The corresponding codeword of weight $8$ in the BCH code $\bC_{(3,26,8,1)}$ is
			$c(x)=x^{22}+x^{16}+x^{14}+x^{13}+x^9+x^3+x+1$.
			Then the BCH code $\bC_{(3,26,8,1)}$ has parameters $[26,11,8]_3$.

			\item Let $p=5$. 
			The polynomial $L(x)=x^{11} + x^{10} + x^9 + x^8 + x^7 + x^6 + x^5 + x^4 + x^3 + x^2 + 4x + 4$ has roots $\{\alpha^2,\alpha^{10},\alpha^{50},\alpha^{250},\alpha^{1250},\alpha^{1562},\alpha^{1564},\alpha^{1572},\alpha^{1612},\alpha^{1812},\alpha^{2812}\}$ in $\F_{5^5}$.
			The corresponding codeword of weight $12$ in the BCH code $\bC_{(5,3124,12,1)}$ is
			$x^{2812} + x^{1812} + x^{1612} + x^{1572} + x^{1564} + x^{1562} + x^{1250} + x^{250} + x^{50} + x^{10} + x^2 + 1
			$.
			Then the BCH code $\bC_{(5,3124,12,1)}$ has parameters $[3124, 3079,12]_5$.
		\end{itemize}
	}
\end{example}

\subsection{$q$-ary BCH codes with $\delta=r \frac{q^m-1}{q-1}+1$}

\begin{proposition}\label{P-2}
	Let $1\le r<q-1$ with $r\mid q-1$, and set
	$t=r\frac{q^m-1}{q-1}.$
	Let $L=\mathbb F_{q^m}^\ast$ and $G(x)=x^t.$
	Then the Goppa code $\Gamma_q(L,G)$ has minimum distance
	$d=t+1.$
\end{proposition}

\begin{proof}
	\rm{
		The proof follows by the same argument as in Theorem \ref{T-4-2}, applied to the monomial Goppa polynomial $G(x)=x^t$.
		Choose $\gamma\in\mathbb F_{q^m}^\ast$ with $\gamma^t\ne 1$, and let
		$$
		F_\gamma(x)=(x+\gamma)^{t+1}-(x+\gamma).
		$$
		By Lemma \ref{L-H}, the polynomial $F_\gamma(x)$ has exactly $t+1$ distinct roots in $L=\mathbb F_{q^m}^\ast$, and for each such root $\alpha$ one has $F_\gamma'(\alpha)\in\mathbb F_q^\ast$. Moreover, for every $\alpha\in L$,
		$$
		G(\alpha)=\alpha^t=N_{\mathbb F_{q^m}/\mathbb F_q}(\alpha)^r\in\mathbb F_q^\ast.
		$$
		Hence the condition of Theorem \ref{T-0} is satisfied, and $\Gamma_q(L,G)$ contains a codeword of weight $t+1$. Therefore $d=t+1$. 
		\qed
	}
\end{proof}

 Combining Lemma \ref{L-dim2} and Proposition \ref{P-2}, we obtain the following result.

\begin{theorem}\label{T-BCH}
	Let $t = r \frac{q^m-1}{q-1}$ with $r \mid q-1$ and $1\le r<q-1$ . The narrow-sense primitive BCH code $\bC_{(q, q^m-1, t+1, 1)}$ has minimum distance $d = t+1$ and dimension $(q-r)^m - 1$.
\end{theorem}

\begin{proof}
	\rm{
		By Lemma \ref{L-2-1}, the Goppa code $\Gamma_q(L,x^t)$ with $L=\mathbb F_{q^m}^\ast$ is equivalent to the primitive BCH code $\bC_{(q,q^m-1,t+1,1)}$. Therefore the minimum distance follows from Proposition \ref{P-2}. The dimension follows from Lemma \ref{L-dim2}.
		\qed
	}
\end{proof}

\begin{example}
	\rm{
		Examples of the codes in Theorem \ref{T-BCH} are given in Table \ref{tab-BCH}. 
		All numerical
		examples were produced using SageMath.
	}
\end{example}

\begin{longtable}{|l|l|l|l|l|}
	\caption{\label{tab-BCH} Codes in Theorem \ref{T-BCH}}\\ 
	\hline 
	$q$ & $m$ & $r$ & $\bC_{(q,q^{m}-1,r\frac{q^m-1}{q-1}+1,1)}$ & $d_{best}$ \\  
	\hline
	\endfirsthead
	
	\hline
	$q$ & $m$ & $r$ & $\bC_{(q,q^{m}-1,r\frac{q^m-1}{q-1}+1,1)}$ & $d_{best}$ \\   & $d_{best}$ \\  
	\hline
	\endhead
	
	\hline
	\endfoot
	
	\multirow{4}{*}{3} 
	& 2 & 1 & $[8,3,5]_3$  &  5 \\ \cline{2-5}
	& 3 & 1 & $[26,7,14]_3$  &  14 \\ \cline{2-5}
	& 4 & 1 & $[80,15,41]_3$  &  41 \\ \cline{2-5}
	& 5 & 1 & $[242,31,122]_3$  & 122  \\ \hline
	
	\multirow{3}{*}{4} 
	& 2 & 1 & $[15,8,6]_4$  &  6 \\ \cline{2-5}
	& 3 & 1 & $[63,26,22]_4$  & 22  \\ \cline{2-5}
	& 4 & 1 & $[255,80,86]_4$  & 86  \\ \hline

	\multirow{4}{*}{5} 
	& \multirow{2}{*}{2} & 1 & $[24,15,7]_5$   &  7 \\ \cline{3-5}
	&  & 2 & $[24,8,13]_5$  & 13  \\ \cline{2-5}
	
	& \multirow{2}{*}{3} & 1 & $[124,63,32]_5$   &  32 \\ \cline{3-5}
	&  & 2 & $[124,26,63]_5$  & 63  \\ \hline
	
	
	\multirow{6}{*}{7} 
	& \multirow{3}{*}{2} & 1 & $[48,35,9]_7$   &  9 \\ \cline{3-5}
	&  & 2 & $[48,24,17]_7$  &  17 \\ \cline{3-5}
	&  & 3 & $[48,15,25]_7$  &  25 \\ \cline{2-5}
	
	& \multirow{3}{*}{3} & 1 & $[342,215,58]_7$   & -  \\ \cline{3-5}
	&  & 2 & $[342,124,115]_7$  & -  \\ \cline{3-5}
	&  & 3 & $[342,63,172]_7$  & -  \\ \hline
	
	\multirow{2}{*}{8} 
	& 2 & 1 & $[63,48,10]_8$  & 10  \\ \cline{2-5}
	& 3 & 1 & $[511,342,74]_8$  & -  \\ \hline
	
	\multirow{3}{*}{9} 
	& \multirow{3}{*}{2} & 1 & $[80,63,11]_9$   &  11 \\ \cline{3-5}
	&  & 2 & $[80,48,21]_9$  & 21  \\ \cline{3-5}
	&  & 4 & $[80,24,41]_9$  & 41  \\ \cline{2-5}
	
	
\end{longtable}

\subsection{$q$-ary BCH codes with $\delta=q^t+1$}

We now turn to the family $\bC_{(q,q^m-1,q^t+1,1)}$. In \cite[Theorem  4.4]{2026BCH}, the same family was proved to have minimum distance $q^t+1$ under the condition $pt \mid m$. 
The following result shows that the approach based on Goppa codes only requires a weaker condition $t \mid m$.

\begin{lemma}\label{L-q^t+1}
	Assume that $t\mid m$ and $t<m$. Let $b\in \mathbb F_{q^m}\setminus \mathbb F_{q^t}$ and $\lambda=\frac{b-b^{q^t}}{b^{q^t}-b^{q^{2t}}}.$
	Define
	$$
	F(x)=x^{q^t+1}+(1+\lambda)x^{q^t}+\lambda\in \mathbb F_{q^m}[x].
	$$
	Then $F(x)$ has exactly $q^t+1$ distinct nonzero roots in $\mathbb F_{q^m}$.

\end{lemma}

\begin{proof}
	\rm{
		Since $b\notin \mathbb F_{q^t}$, we have $b^{q^t}\neq b$, and hence $b^{q^t}-b^{q^{2t}}\neq 0$.
		Moreover,
		$$
		F(x)=\frac{(b^{q^t}x+b)(x+1)^{q^t}-(b^{q^t}x+b)^{q^t}(x+1)}
		{b^{q^t}-b^{q^{2t}}}.
		$$
		A direct substitution shows that $x=-1$ is a root of $F(x)$. For $x\neq -1$,
		$$
		F(x)=0
		\iff
		\left(\frac{b^{q^t}x+b}{x+1}\right)^{q^t}
		=\frac{b^{q^t}x+b}{x+1},
		$$
		which is equivalent to
		$$
		\frac{b^{q^t}x+b}{x+1}\in \mathbb F_{q^t}.
		$$
		Thus, for each $u\in \mathbb F_{q^t}$, the roots with $x\neq -1$ are given by
		$$
		x=\frac{u-b}{b^{q^t}-u}.
		$$
		This is well defined since $b^{q^t}\notin \mathbb F_{q^t}$. Moreover, if
		$$
		\frac{u_1-b}{b^{q^t}-u_1}=\frac{u_2-b}{b^{q^t}-u_2},
		$$
		then $u_1=u_2$, and hence the correspondence is injective. Therefore, $F(x)$ has $q^t$ distinct roots with $x\neq -1$.
		Since $\deg(F)=q^t+1$, it follows that these together with $x=-1$ give all the roots of $F(x)$ in $\mathbb F_{q^m}$.
		Finally, $F(0)=\lambda\neq 0$, so all roots are nonzero.
		\qed}
\end{proof}
\begin{theorem}\label{T-q^t+1}
	Let $q$ be a prime power. For any $t \mid m$ and $t<m$, the narrow-sense primitive BCH code $\bC_{(q, q^m-1, q^t+1, 1)}$ has minimum distance $d = q^t+1$. Moreover, if the dimension formula in Lemma~\ref{L-dim1} applies, then
	$k=q^m-1-mq^{t-1}(q-1).$
\end{theorem}

\begin{proof}
\rm{
	The dimension follows from Lemma \ref{L-dim1}.
	Let $\alpha$ be a primitive element of $\F_{q^m}$. Define
	$L=\{1,\alpha,\alpha^2,\dots,\alpha^{q^m-2}\}$ and $G(x)=x^{q^t}.$
	Choose $b\in \F_{q^m}\setminus \F_{q^t}$ and set
	$$
	F(x)=x^{q^t+1}+(1+\lambda)x^{q^t}+\lambda,
	$$
	with $\lambda=\frac{b-b^{q^t}}{b^{q^t}-b^{q^{2t}}}$.
	By Lemma \ref{L-q^t+1}, $F(x)$ has exactly $q^t+1$ distinct nonzero roots in $\F_{q^m}$. Since all roots are nonzero, they lie in $L=\F_{q^m}^\ast$.
	
	Since $F'(x)=x^{q^t}=G(x)$, it follows from Corollary \ref{C-derivative} that the minimum distance of $\Gamma_q(L,G)$ is $q^t+1$.
	By Lemma \ref{L-2-1}, the Goppa code $\Gamma_q(L,G)$ is equivalent to the primitive BCH code $\bC_{(q,q^m-1,q^t+1,1)}$. Hence the BCH code $\bC_{(q,q^m-1,q^t+1,1)}$ also has minimum distance $q^t+1$.
	\qed
}
\end{proof}

\begin{remark}
\rm{
	Under the equivalence in Lemma \ref{L-2-1}, if a codeword of $\Gamma_q(L,G)$ has support $\{i_1,\dots,i_{q^t+1}\}$, then the corresponding BCH codeword has support
	$\{-i_1,-i_2,\dots,-i_{q^t+1}\}\pmod n.$
}
\end{remark}

\begin{example}
	\rm{

		In the following examples, we take $L=\{1,\alpha,\alpha^2,\dots,\alpha^{q^m-2}\}$ and $G(x)=x^{q^t}$, where $\alpha$ is a primitive element of $\F_{q^m}$, and choose $b=\alpha$. All computations were confirmed by SageMath.
		\begin{itemize}
			\item Let $(q,t,m)=(3,1,5)$. The polynomial
			$F(x)=x^4 + (\alpha^4 + 2\alpha^2 + 2)x^3 + \alpha^4 + 2\alpha^2 + 1$ has roots $\{\alpha^{111}, \alpha^{119}, \alpha^{121}, \alpha^{225}\}$ in $\F_{3^5}$.
			This yields a Goppa codeword
			$c(x)=x^{111}+x^{119}+x^{121}+x^{225}$
			of weight $4$. The corresponding codeword in $\bC_{(3,242,4,1)}$ is
			$c'(x)=x^{131}+x^{123}+x^{121}+x^{17}.$
			Hence $\bC_{(3,242,4,1)}$ has parameters $[242,232,4]_3$.
			
			\item Let $(q,t,m)=(4,1,3)$. The polynomial
			$F(x)=x^5 + (\alpha^5 + \alpha^4 + \alpha^2)x^4 + \alpha^5 + \alpha^4 + \alpha^2 + 1$ has roots $\{1, \alpha^{6}, \alpha^{36}, \alpha^{45}, \alpha^{48}\}$.
			This yields a Goppa codeword $c(x)=x^{48} + x^{45} + x^{36} + x^{6} + 1$ of weight $5$.
			The corresponding codeword in $\bC_{(4,63,5,1)}$ is $c'(x)=x^{63} + x^{57} + x^{27} + x^{18} + x^{15}$. Hence $\bC_{(4,63,5,1)}$ has parameters $[63,54,5]_4$.
									
			\item Let $(q,t,m)=(5,1,3)$. The polynomial $F(x)=x^6 + (4\alpha^2 + 3\alpha + 4)x^5 + 4\alpha^2 + 3\alpha + 3$ has roots $\{\alpha^{2}, \alpha^{30}, \alpha^{38}, \alpha^{42}, \alpha^{62}, \alpha^{110}\}$ in $\F_{5^3}$.
			This yields a Goppa codeword $c(x)=x^{110} + x^{62} + x^{42} + x^{38} + x^{30} + x^2$ of weight $6$.
			The corresponding codeword in $\bC_{(5, 124, 6, 1)}$ is
			$c'(x)=x^{122} + x^{94} + x^{86} + x^{82} + x^{62} + x^{14}$.
			Hence $\bC_{(5, 124, 6, 1)}$ has parameters $[124,112,6]_5$.
			
		\end{itemize}
	}
\end{example}

\section{Summary and concluding remarks}\label{s6}

In this paper, we determined the minimum distances of several families of $q$-ary Goppa codes and primitive BCH codes. The main theoretical contribution is a necessary and sufficient criterion for a Goppa code to attain its designed distance $\delta=t+1$. As applications, we obtained three families of Goppa codes with exact minimum distance. In particular, we proved the tightness of the lower bound for wild Goppa codes and extended the family of Goppa codes with $G(x)=x^t+A$ from the binary case to arbitrary odd prime powers.

A further contribution is the specialization of this criterion to the monomial case $G(x)=x^t$, which corresponds to the primitive BCH codes with designed distance $t+1$. Under this specialization, we recovered the locator criterion in \cite[Theorem 3.1]{2026BCH} from the Goppa viewpoint. This led to four families of primitive BCH codes whose minimum distances are equal to the designed distances, including the binary codes $\bC_{(2,2^m-1,\delta,1)}$ with $\delta \in \{9,15\}$, the family $\bC_{(p,p^p-1,2p+2,1)}$, the family $\bC_{(q,q^m-1,r\frac{q^m-1}{q-1}+1,1)}$, and the family $\bC_{(q,q^m-1,q^t+1,1)}$. In particular, for the BCH code $\bC_{(q,q^m-1,q^t+1,1)}$, we improved the condition $pt \mid m$ in \cite[Theorem 4.4]{2026BCH} to the weaker assumption $t\mid m$.

These results show that the Goppa framework provides a natural and effective approach to determining minimum distances. It remains of interest to extend this approach to other classes of alternant codes and to additional families of BCH codes with unknown minimum distances.

\end{document}